\documentclass[12pt, draftclsnofoot, onecolumn]{IEEEtran}
\usepackage{graphicx,amsmath,amssymb}
\usepackage{cite}
\usepackage{subfigure}
\usepackage{citesort}
\usepackage{fancyhdr}
\usepackage{mdwmath}
\usepackage{mdwtab}
\usepackage{balance}
\usepackage{xcolor}
\usepackage{bm}
\usepackage{amsthm}
\usepackage{algorithm}
\usepackage{algorithmic}
\usepackage{multirow}
\usepackage{flafter}
\usepackage{epstopdf}
\newtheorem{remark}{Remark}

\newtheorem{theorem}{Theorem}

\newtheorem{lemma}{Lemma}

\newtheorem{corollary}{Corollary}

\newtheorem{assumption}{Assumption}

\begin{document}
\title{UAV-Aided Multi-Way NOMA Networks with Residual Hardware Impairments}

\author{Xingwang~Li,~\IEEEmembership{Senior Member,~IEEE,}
        Qunshu~Wang,~\IEEEmembership{Student Member,~IEEE,}
         Yuanwei~Liu,~\IEEEmembership{Senior Member,~IEEE,}
         Theodoros~A.~Tsiftsis,~\IEEEmembership{Senior,~IEEE,}
         Zhiguo~Ding,~\IEEEmembership{Fellow,~IEEE,}
        and Arumugam~Nallanathan,~\IEEEmembership{Fellow,~IEEE}
\thanks{X. Li and Q. Wang are with the School of Physics and Electronic Information Engineering, Henan Polytechnic University, Jiaozuo, China (email:lixingwangbupt@gmail.com, 15333766153@163.com).}
\thanks{Y. Liu and A. Nallanathanis are with the School of Electronic Engineering and Computer Science, Queen Mary University of London, London, UK (email:\{yuanwei.liu, a.nallanathan\}@qmul.ac.uk).}
\thanks{T. A. Tsiftsis is with the School of Intelligent Systems Science and
Engineering, Jinan University Zhuhai Campus, Zhuhai 519070, China (email: theo\_tsiftsis@jnu.edu.cn).}
\thanks{Z. Ding is with the School of Electrical and Electronic Engineering, The University of Manchester, Manchester, UK (email: zhiguo.ding@manchester.ac.uk).}
}
\maketitle
\nocite{Li2018}
\begin{abstract}
In this paper, we study an unmanned aerial vehicle (UAV)-aided non-orthogonal multiple access (NOMA) multi-way relaying networks (MWRNs). Multiple terrestrial users aim to exchange their mutual information via an amplify-and-forward (AF) UAV relay. Specifically, the realistic assumption of the residual hardware impairments (RHIs) at the transceivers is taken into account. To evaluate the performance of the considered networks, we derive the analytical expressions for the achievable sum-rate (ASR). In addition, we carry out the asymptotic analysis by invoking the affine expansion of the ASR in terms of \emph{high signal-to-noise ratio (SNR) slope} and \emph{high SNR power offset}. Numerical results show that: 1) Compared with orthogonal multiple access (OMA), the proposed networks can significantly improve the ASR since it can reduce the time slots from $\left[ {\left( {M - 1} \right)/2} \right] + 1$ to 2; and 2) RHIs of both transmitter and receiver have the same effects on the ASR of the considered networks.

\end{abstract}

\begin{IEEEkeywords}
Hardware impairments, multi-way relay, NOMA, UAV
\end{IEEEkeywords}

\section{Introduction}
\IEEEPARstart {N}{on}-orthogonal multiple access (NOMA) has attracted significant attentions from academia and industry for the fifth generation (5G) and beyond mobile communication networks since it can provide high spectrum efficiency (SE), massive connectivity and user fairness \cite{Z.Ding.2017}. Compared with current orthogonal multiple access (OMA), the main advantage of NOMA is that it can allow multiple users to access the same time/frequency resources by power multiplexing\textcolor[rgb]{0.00,0.00,1.00}{ \footnote{Normally, NOMA can be classified into  two categories, namely code-domain NOMA and power-domain NOMA. The analysis of this paper is based on power-domain NOMA.}}. In addition, NOMA can ensure user fairness by allocating more power to the weak channel condition users \cite{Li.X2019}.

On a parallel avenue, unmanned aerial vehicle (UAV) communication has been in the limelight for the important supplement of current terrestrial communication because of flexibility and higher maneuverability \cite{8660516LiuTWC},\cite{8710357Zhang}. Contrary to the terrestrial wireless networks, UAV can provide emergency communication service in the cases of natural disaster when the terrestrial networks are partially or entirely malfunctioning \cite{8660516Mozaffari}. Moreover, UAV can be acted as a UAV-mounted base station (BS) to provide short-term erratic high rate wireless connectivity of the terrestrial users or a relay to enhance the quality-of-service (QoS) of cell-edge users. Compared with the terrestrial BS or relay, the most prominent advantage of the UAV-mounted BS or relay is that it can rapidly be deployed in any place\cite{7486987Mozaffari}. In addition, UAV-aided relay can obviously provide wider coverage range than other relays. These benefits can be achieved since it is not only ease of deployment but also high QoS caused by line-of-sight (LoS) links \cite{8951059XingwangIA}.

To reap the benefits of both NOMA and UAV, NOMA has been introduced to improve the performance of UAV-enabled networks, which has sparked a great deal of research interests \cite{8951059XingwangIA,8269066Sharma,8488592Hou,8663350ZhiguoDing}. Considering the satellite-terrestrial communication, authors in \cite{8951059XingwangIA} proposed a unified framework of hybrid satellite/UAV-assisted NOMA networks, where the optimal location scheme of a UAV was designed. The UAV acting as the flying BS to communicate with two ground NOMA users, authors of \cite{8269066Sharma} highlighted the outage comparsion of the NOMA and OMA schemes. Exploiting stochastic geometry, authors in \cite{8488592Hou} proposed a 3-D UAV framework for randomly roaming NOMA users, and the outage probability (OP) and the ergodic rate of the proposed framework were derived. Liu \emph{et al.} in \cite{8663350ZhiguoDing} proposed a joint optimization scheme for the placement and power allocation to maximize the total ASR of NOMA-UAV network.
\begin{figure}[!t]
\setlength{\abovecaptionskip}{0pt}
\centering
\includegraphics [width=2.6in]{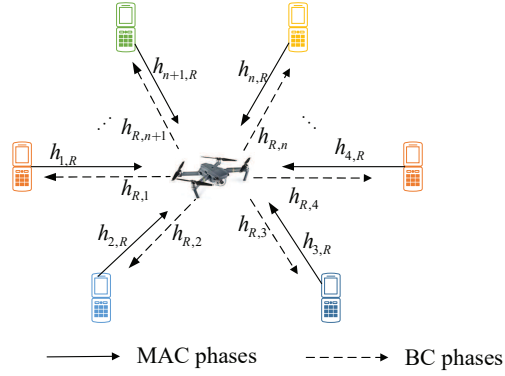}
\caption{System model}
\label{fig1}
\end{figure}

Meanwhile, multi-way relaying protocol is considered as another promising technology to further improve SE. For multi-way relaying protocol, a set of users aim to exchange their mutual information with the aid of a relay \cite{J.Xue2015}. In \cite{G.Amarasuriya2012}, Amarasuriya \emph{et al.} proposed a pairwise amplify-and-forward (AF) multi-way relay networks (MWRNs) and evaluated the OP and the average bit-error rate of the considered network. \cite{6809883Sharifian} investigated the performance of complex field network coding in MWRNs and considered full data exchange algorithm in order to improve the throughput. In \cite{C.D.Ho2018}, the combination between MWRNs and massive multiple-input multiple-output (MIMO) technology was accomplished and a novel strategy for massive MIMO systems with successive cancelation decoding was designed. The common characteristic of the aforementioned works is based on OMA, which needs more time slots to complete the information exchange. To further improve SE, Silva \emph{et al.} in \cite{S.Silva2019} employed NOMA technology in multi-way massive MIMO relay networks, which shows that it does not only exchange information among users, but also reduce the processing time. Thus, it is of profound importance to investigate the performance of UAV-aided NOMA MWRNs.


Motivated by the above discussion, we investigate the performance of a UAV-aided NOMA MWRNs. More practical, we assume that all nodes suffer from the effects of residual hardware impairments (RHIs), which caused by time-varying hardware characteristics and imperfect compensation algorithms \cite{X.Li2019},\cite{8892558Li}. To the best of our knowledge, a detailed analysis of UAV-assisted NOMA MWRNs is missing in the literature. To fill this gap, we analyze the impact of RHIs on the performance of UAV-aided NOMA MWRNs by deriving the approximate analytical expressions for the  achievable sum rate (ASR). To obtain more insights, the asymptotic analysis in the high signal-to-noise-ratio (SNR) regions is carried out by invoking \emph{high SNR slope} and \emph{high SNR power offset}.

\section{System Model}\label{sec2}


We consider a UAV-aided NOMA MWRN as illustrated in Fig. 1, where $M$ ground users mutually exchange their information with the aid of an AF UAV relay $R$.
We assume that: 1) each node is equipped with a single antenna; 2) both UAV and NOMA users operate in the half-duplex mode; 3) the direct links among all NOMA users are deemed to nonexistent owing to severe shadowing fading, and 4) this system operates in time division duplex (TDD). This means that there is a reciprocity between uplink and downlink channels.

The channel between the $i$-th user ($i = 1,2, \ldots ,M$) and UAV is denoted by ${g_i}$, and ${\left| {{g_i}} \right|^2} = \frac{{{{\left| {{h_i}} \right|}^2}}}{{1 + d_i^\upsilon }}$, where ${d_i}$ is the distance between $R$ and ${S_i}$, and $\nu $ is the path-loss component. We assume that ${h_i}$ follows the general Nakagami-$m$ distribution. Without loss of generality, it is assumed that ${\left| {{g_1}} \right|^2} \le {\left| {{g_2}} \right|^2} \le   \ldots  \le {\left| {{g_M}} \right|^2}$. The probability density function (PDF) and the cumulative distribution function (CDF) of ${\left| {{h_i}} \right|^2}$ are assumed to be the same as in \cite{7870605Yue}.

The entire communication process can be divided into two phases:  multiple-access (MAC) and broadcast (BC).

\emph{MAC:} Each user $S_i$, $i = 1,2, \ldots ,M$, sends its signal $t = \sqrt {{a_i}{P}}{x_i} $ to $R$, where ${{P}}$ is the allowable maximum
transmit power for each user;\footnote{In this study, we assume that all users have the equal maximum power. Note that some power optimization schemes are capable of further improving the system performance, though this is part of our ongoing work.} ${x_i}$ is the transmitted signal from the $i$-th user with $E\left\{ {{{\left| {{x_i}} \right|}^2}} \right\} = 1$, where $E\{\cdot\}$ denotes the expectation operator of random variables; ${{a_i}}$ denotes power allocation coefficient between the UAV and $i$-th user with ${a_1} + {a_2} + \ldots  + {a_M} = 1$ and ${a_1} > a_2 > \ldots  > {a_M}$. Thus, the received superposed signal of $R$ can be expressed as
\begin{align}\label{5}
{y_r} = \sum\limits_{i = 1}^M {{g_i}\left( {{x_i}\sqrt {{a_i}{P}}  + {\eta _{{u_{t,i}}}}} \right)}  + {\eta _{Rr}} + {N_r},
\end{align}
where ${\eta _{{u_{t,i}}}}$ and ${\eta _{Rr}}$ are the distortion noises from the transmitter and receiver with
${\eta _{{u_{t,i}}}} \sim \mathcal{CN}\left( {0,\kappa _{{u_{t,i}}}^2{a_i}P} \right)$ and ${\eta _{Rr}} \sim \mathcal{CN}\left( {0,\kappa _{Rr}^2\sum_{i = 1}^M {{{\left| {{g_i}} \right|}^2}{a_i}} P} \right)$, respectively; ${N_r}\sim \mathcal{CN}\left( {0,\sigma _r^2}\right)$ is the additive complex white Gaussian noise (AWGN).

\emph{BC:} $R$ amplifies and forwards received superposed signal to all users simultaneously. Thus, the received signal at the $k$-th user is given as
\begin{align}\label{6}
{y_{r,k}} = {g_k}\left( {G{y_r} + {\eta _{Rt}}} \right) + {\eta _{{u_{r,k}}}} + {N_{t,k}},
\end{align}
where $G = \sqrt {{P_R}/\left( {\sum\nolimits_{i = 1}^M {{\rho _i}{a_i}P\left( {1 + \kappa _{{u_{t,i}}}^2 + \kappa _{Rr}^2} \right)}  + \sigma _r^2} \right)} $ is the normalization coefficient and ${{P_R}}$ is the transmitted power at $R$; ${\rho _i} = {\left| {{g_i}} \right|^2}$ is the channel gain; ${\eta _{{u_{r,k}}}}$ and ${\eta _{Rt}}$ represent the distortion noises from the receiver and transmitter with ${\eta _{{u_{r,k}}}} \sim \mathcal{CN}\left( {0,\kappa _{{u_{r,k}}}^2{\rho _k}{P_R}} \right)$ and ${\eta _{Rt}} \sim \mathcal{CN}\left( {0,\kappa _{Rt}^2{P_R}} \right)$, respectively; ${N_{t,k}} \sim \mathcal{CN}\left( {0,\sigma _{t, k}^2} \right)$ denotes the AWGN. In this phase, the decoding order is carried out from user with a worse channel condition to user with a better channel condition.

According to \cite{S.Silva2019}, the received signal at ${S_k}$ when ${S_k}$ decodes $n$-th user after decoding ${n-1}$ users should be divided into two cases:
\subsubsection{The first case}
When $n \ge k$, the received signal at ${S_k}$ can be written as
\begin{align}\nonumber
&{y_{k,n}} \!=\! {g_k}G{g_{n + 1}}{x_{n + 1}}\sqrt {{a_{n + 1}}{P}} \! +\! {g_k}G\sum\limits_{i = n + 2}^M {{g_i}{x_i}\sqrt {{a_i}{P}} } \\
& \!+\! {g_k}G\!\sum\limits_{i = 1}^M {{g_i}{\eta _{{u_{t,i}}}} \!\!\! + \!{g_k}G\left({\eta _{Rr}} \!+ \!{N_r}\right) \!+\! {g_k}{\eta _{Rt}}\! + \!{\eta _{{u_{r,k}}}} \!\!\!+\! {N_{t,k}}}.
\end{align}

For this case, ${S_k}$ cannot successfully decode ${S_n}$ and the received signal-to-interference-plus-noise ratio (SINR) is equal to zero. Obviously, the ASR of this case is equal to zero, which is omitted in the sequel.
\subsubsection{The second case}
When $n < k$, the received signal at ${S_k}$ can be written as
\begin{align}\nonumber
{\noindent y_{k,n}} = &{g_k}G{g_n}{x_n}\sqrt {{a_n}{P}}  + {g_k}G\sum\limits_{i = n + 1}^{M - 1} {{g_i}{x_i}\sqrt {{a_i}{P}} } + {g_k}G{N_r}\\
&+ {g_k}G\sum\limits_{i = 1}^M {{g_i}{\eta _{{u_{t,i}}}} \!\!\! + \!{g_k}G{\eta _{Rr}} \! +\! {g_k}{\eta _{Rt}}\! +\! {\eta _{{u_{r,k}}}}\!\!\! + \!{N_{t,k}}}.
\end{align}

For this case, ${S_k}$ can successfully decode ${S_n}$ and the received SINR is expressed as
\begin{align}\label{9}
\gamma_{{k,n}} = \frac{{{\rho _k}{\rho _n}{a_n}{r_1}{r_2}}}{{{\Theta _1} + {\Theta _2} + {\Theta _3} + {\Theta _4} + {\Theta _5}}},
\end{align}
where ${r_1}\! = \!\frac{{{P}}}{{\sigma _r^2}}$ and ${r_2} \!=\! \frac{{{P_R}}}{{\sigma _{t,k}^2}}$ are the transmit SNRs at NOMA users and $R$, respectively; ${\Theta _1} = \sum_{i = n + 1}^{M - 1} {{\rho _k}{\rho _i}{a_i}{r_1}{r_2}} $,\\ ${\Theta _2} = \sum\nolimits_{i = 1}^M\! {{\rho _k}{\rho _i}{r_1}{r_2}\left( {{a_i}\kappa _{{u_{t,i}}}^2\!\!\! +\!  {a_i}\kappa _{Rr}^2\! +\! \kappa _{Rt}^2\kappa _{Rr}^2\! +\!\! \kappa _{Rt}^2{a_i}\! +\!\!
{a_i}\kappa _{Rt}^2\kappa _{{u_{t,i}}}^2} \right)} $,\\ ${\Theta _3}\! \!=\! \sum\nolimits_{i = 1}^M\! {{a_i}\rho _i^2\kappa _{{u_{r,k}}}^2\!\!{r_1}{r_2}\! +\!\! {a_i}\rho _i^2\kappa _{{u_{r,k}}}^2\!\!{r_1}{r_2}\kappa _{{u_{t,i}}}^2\!\!\! +\!\! {a_i}\rho _i^2\kappa _{{u_{r,k}}}^2{r_1}{r_2}\kappa _{Rr}^2} $,\\ ${\Theta _4} = \sum\nolimits_{i = 1}^M {\left( {{\rho _i}{r_1}{a_i} + {\rho _i}{r_1}{a_i}\kappa _{{u_{t,i}}}^2 + {\rho _i}{r_1}{a_i}\kappa _{Rr}^2 + \kappa _{{u_{r,k}}}^2{\rho _i}{r_2}} \right)} $,\\ ${\Theta _5} = {\rho _k}{r_2} \!+\! {\rho _k}\kappa _{Rt}^2{r_2} \!+\!\! 1$.

\section{Achievable Sum Rate Analysis}
In this section, we investigate the analytical expression of the ASR of UAV-aided NOMA MWRNs. Then the asymptotic analysis is explored in the high SNR regime.
\subsection{Achievable Sum Rate Analysis}
Based on (\ref{9}), the achievable rate of $k$-th user is given as
\begin{align}\label{10}
{C_{k,n}}= E\left[ {\frac{1}{2}\log_2\left( {1 + \gamma_{{k,n}}} \right)} \right],
\end{align}
where the constant $1/2$ represents that the communication process is completed into two time slots. The following theorem presents the achievable rate of the considered system in the presence of either ideal or non-ideal conditions.

\begin{theorem}
For ideal/non-ideal conditions, we have

$\bullet$ Non-ideal conditions ($\kappa _{{u_{t,i}}} \neq 0, \kappa _{{u_{r,i}}} \neq 0,  \kappa _{Rr} \neq 0, \kappa _{Rt} \neq 0, 0 \leq i \leq M$)
\begin{align}\label{11}
C_{k,n}^{ni} = \frac{1}{2}\log 2\left( {1 + \frac{{{\Psi _k}{\Psi _n}{a_n}{r_1}{r_2}}}{{{\Xi _1} + {\Xi _2} + {\Xi _3} + {\Xi _4} + {\Xi _5}}}} \right),
\end{align}
where ${\Xi _1} = \sum_{i = n + 1}^{M - 1} {{\Psi _k}{\Psi _i}{a_i}{r_1}{r_2}}$, \\ ${\Xi _2} = \sum\limits_{i = 1}^M {{\Psi _k}{\Psi _i}\kappa _{{u_{t,i}}}^2{a_i}{r_1}{r_2}} + \kappa _{Rr}^2\sum\limits_{i = 1}^M {{\Psi _k}{\Psi _i}{a_i}{r_1}{r_2}} $,\\ ${\Xi _3}\! =\! \sum\limits_{i = 1}^M {{\Psi _k}{\Psi _i}\kappa _{Rt}^2{r_1}{r_2}\left( {{a_i} + {a_i}\kappa _{{u_{t,i}}}^2} \right)} + \sum\limits_{i = 1}^M {{\Psi _k}{\Psi _i}\kappa _{Rt}^2\kappa _{Rr}^2{a_i}{r_1}{r_2}} $,\\ ${\Xi _4} = \sum\limits_{i = 1}^M {{\Omega _i}\kappa _{{u_{r,k}}}^2{r_1}{r_2}\left( {{a_i} + {a_i}\kappa _{{u_{t,i}}}^2} \right)} + \sum\limits_{i = 1}^M {{\Omega _i}\kappa _{{u_{r,k}}}^2\kappa _{Rr}^2{a_i}{r_1}{r_2}}  + \sum\limits_{i = 1}^M {{\Psi _i}{r_1}\left( {{a_i} + {a_i}\kappa _{{u_{t,i}}}^2} \right)}  + \sum\limits_{i = 1}^M {{\Psi _i}\kappa _{Rr}^2{a_i}{r_1}}  $,\\ ${\Xi _5} = \sum_{i = 1}^M {\kappa _{{u_{r,k}}}^2{\Psi _i}{r_2}}  + {\Psi _k}{r_2} + {\Psi _k}\kappa _{Rt}^2{r_2} + 1$,\\  ${\Psi _i}$ and ${\Omega _i}$ are given as
\begin{align}
{\Psi _i} =& \frac{1}{{1 + d_i^\nu }}\frac{1}{{\Gamma \left( {{\alpha _i}} \right)\beta _i^{{\alpha _i}}}}\sum\limits_{n = 0}^{m - 1} {\left( {\begin{array}{*{20}{c}}
{m - 1}\\
n
\end{array}} \right)} {\left( { - 1} \right)^n}\sum\limits_{{p_0} +  \ldots  + {p_{{\alpha _i} - 1}} = n + M - m} {\left( {\begin{array}{*{20}{c}}
{n + M - m}\\
{{p_0}, \ldots ,{p_{{\alpha _i} - 1}}}
\end{array}} \right)} \\ \nonumber
&\frac{{M!}}{{\left( {m - 1} \right)!\left( {M - m} \right)!}}{\prod\limits_{{g_i} = 0}^{{\alpha _i} - 1} {\left( {\frac{1}{{{g_i}!\beta _i^{{g_i}}}}} \right)} ^{{p_{{g_i}}}}}\left( {{\alpha _i} + {g_i}{p_{{g_i}}}} \right)!{\left( {\frac{{\left( {n + M + 1 - m} \right)}}{{{\beta _i}}}} \right)^{ - \left( {{\alpha _i} + {g_i}{p_{{g_i}}}} \right) - 1}}.
\end{align}
\begin{align}
{\Omega _i} =&\frac{1}{{1 + d_i^\nu }} \frac{1}{{\Gamma \left( {{\alpha _i}} \right)\beta _i^{{\alpha _i}}}}\sum\limits_{n = 0}^{m - 1} {\left( {\begin{array}{*{20}{c}}
{m - 1}\\
n
\end{array}} \right)} {\left( { - 1} \right)^n}\sum\limits_{{p_0} +  \ldots  + {p_{{\alpha _i} - 1}} = n + M - m} {\left( {\begin{array}{*{20}{c}}
{n + M - m}\\
{{p_0}, \ldots ,{p_{{\alpha _i} - 1}}}
\end{array}} \right)} \\\nonumber
&\frac{{M!}}{{\left( {m - 1} \right)!\left( {M - m} \right)!}}
 \times {\prod\limits_{{g_i} = 0}^{{\alpha _i} - 1} {\left( {\frac{1}{{{g_i}!\beta _i^{{g_i}}}}} \right)} ^{{p_{{g_i}}}}}\left( {{\alpha _i} + 1 + {g_i}{p_{{g_i}}}} \right)!{\left( {\frac{{\left( {n + M + 1 - m} \right)}}{{{\beta _i}}}} \right)^{ - \left( {{\alpha _i} + 1 + {g_i}{p_{{g_i}}}} \right) - 1}}.
\end{align}

$\bullet$ Ideal conditions  ($\kappa _{{u_{t,i}}} = \kappa _{{u_{r,i}}} = \kappa _{Rr} = \kappa _{Rt} = 0, 0 \leq i \leq M$)
\begin{align}
C_{k,n}^{id} = \frac{1}{2}\log 2\left( {1 + \frac{{{\Psi _k}{\Psi _n}{a_n}{r_1}{r_2}}}{{\varpi  + {\Psi _k}{r_2} + 1}}} \right),
\end{align}
where $\varpi  = \sum_{i = n + 1}^{M - 1} {{\Psi _k}{\Psi _i}{a_i}{r_1}{r_2}}  + \sum_{i = 1}^M {{\Psi _i}{r_1}{a_i}} $.

\end{theorem}
\begin{proof}
Substituting (5) into (6), utilizing Jensen's inequality, the achievable rate in the presence of non-ideal conditions can be obtained after some algebraic manipulations. For ideal conditions, we just need to set ${\kappa _{{u_{t,i}}}} = {\kappa _{Rt}} = {\kappa _{Rr}} = {\kappa _{{u_{r,k}}}} = 0$ in (7).
\end{proof}

\begin{corollary}
For ideal/non-ideal conditions, we have

$\bullet$ Non-ideal conditions

The ASR of UAV-aided NOMA MWRNs in the presence of non-ideal conditions is given by
\begin{align}
{\Upsilon _{ni}} = \sum\limits_{k = 1}^M {\sum\limits_{n = 1}^{M - 1} {C_{k,n}^{ni}} } .
\end{align}

$\bullet$ Ideal conditions

The ASR of UAV-aided NOMA MWRNs in the presence of ideal conditions is given by
\begin{align}
{\Upsilon _{id}} = \sum\limits_{k = 1}^M {\sum\limits_{n = 1}^{M - 1} {C_{k,n}^{id}} } .
\end{align}
\end{corollary}

\subsection{High SNR Analysis}
The high SNR region refers to ${r_1} \to \infty $ and ${r_2} = c r_1 $, where $c$ is a fixed constant. The asymptotic analysis is provided in the following corollaries.
\begin{corollary}
For ideal/non-ideal conditions, we have

$\bullet$ Non-ideal conditions
\begin{align}
C_{k,n}^{ni,\infty } = \frac{1}{2}\log 2\left( {1 + \frac{{{\Psi _k}{\Psi _n}{a_n}}}{{{\Delta _1} + {\Delta _2} + {\Delta _3} + {\Delta _4}}}} \right),
\end{align}
where $~~~{\Delta _1} ~=~ \sum_{i = n + 1}^{M - 1}~~ {{\Psi _k}~~{\Psi _i}~~{a_i}} $, $~~~{\Delta _2} ~=~ \sum\limits_{i = 1}^M ~~{{\Psi _k}~~{\Psi _i}~~{a_i}~~\left( {\kappa _{{u_{t,i}}}^2 ~~+~~ \kappa _{Rr}^2} \right)} $, \\${\Delta _3} = \sum_{i = 1}^M {{\Psi _k}{\Psi _i}\kappa _{Rt}^2\left( {{a_i} + {a_i}\kappa _{{u_{t,i}}}^2{\rm{ + }}{a_i}\kappa _{Rr}^2} \right)} $, ${\Delta _4} = \sum_{i = 1}^M {{\Omega _i}\kappa _{{u_{r,k}}}^2\left( {{a_i} + {a_i}\kappa _{{u_{t,i}}}^2{\rm{ + }}{a_i}\kappa _{Rr}^2} \right)} $.

$\bullet$ Ideal conditions
\begin{align}
C_{k,n}^{id,\infty } = \frac{1}{2}\log 2\left( {1 + \frac{{{a_n}{\Psi _n}}}{{\sum\limits_{i = n + 1}^{M - 1} {{a_i}{\Psi _i}} }}} \right).
\end{align}
\end{corollary}

\begin{corollary}
The total achievable sum rate in the high SNR region can be expressed.

$\bullet$ Non-ideal conditions
\begin{align}
\Upsilon _{ni}^\infty  = \sum\limits_{k = 1}^M {\sum\limits_{n = 1}^{M - 1} {C_{k,n}^{ni,\infty }} },
\end{align}

$\bullet$ Ideal conditions
\begin{align}
\Upsilon _{id}^\infty  = \sum\limits_{k = 1}^M {\sum\limits_{n = 1}^{M - 1} {C_{k,n}^{id,\infty }} } .
\end{align}
\end{corollary}
\begin{remark}
From Corollary 1 to Corollary 3, the analytical expressions for the ASRs of UAV-aided NOMA MWRNs and the asymptotic behaviors in the high SNR regimes are obtained.
\end{remark}

To further obtain ASR performance in the high SNR region, we focus on the \emph{high SNR slope} and \emph{high SNR power offset}. According to \cite{X.Li2019}, the asymptotic achievable rate can be formulated in a general form as
\begin{align}
C_{sum}^\infty  = {\mathcal{S}_\infty }\left( {{{\log }_2}\ r  - {\mathcal{L}_\infty }} \right) +  \circ \left( 1 \right),
\end{align}
where $\circ(\cdot)$ represents the infinitesimal of lower order of x, $r \in \left\{ {{r_1},{r_2}} \right\}$ is the average SNR, ${{\mathcal{S}_\infty }}$ and ${\mathcal{L}_\infty }$ are the high SNR slope in bits/s/HZ (3 dB) and the high SNR power offset (3dB) units, respectively. Based on \cite{X.Li2019}, we can obtain
\begin{align}
{\mathcal{S}_\infty }{\rm{ = }}\mathop {\lim }\limits_{\ r  \to \infty } \frac{{C_{sum}^\infty }}{{{{\log }_2}\ r }},
\end{align}
\begin{align}
{\mathcal{L}_\infty }{\rm{ = }}\mathop {\lim }\limits_{\ r  \to \infty } \left( {{{\log }_2}\ r  - \frac{{C_{sum}^\infty }}{{{S_\infty }}}} \right).
\end{align}
Based on the above definitions, the \emph{high SNR slope} and the \emph{high SNR power offset} can be evaluated in the following corollary.
\begin{corollary}
The high SNR slope and the high SNR power offset are obtained in two condition:

$\bullet$ Non-ideal conditions
\begin{align}
{\mathcal{S}_ {ni}^\infty } = 0,{\mathcal{L}_{ni}^\infty } = \infty.
\end{align}

$\bullet$ Ideal conditions
\begin{align}
{\mathcal{S}_{id}^\infty } = 0,{\mathcal{L}_{id}^\infty } = \infty.
\end{align}
\end{corollary}
\begin{remark}
From Corollary 4, one can obtain that the high SNR slope and high SNR power offset in two conditions equal to zero and infinity, respectively. This happens due to the fact that an error floor exists in high SNR regions owing to the interference from other users.
\end{remark}

\begin{figure}[t!]
\centering
\subfigure[{OMA/NOMA}]{
\begin{minipage}[t]{0.46\linewidth}
\centering
\includegraphics[width= 1.9in, height=2.3in]{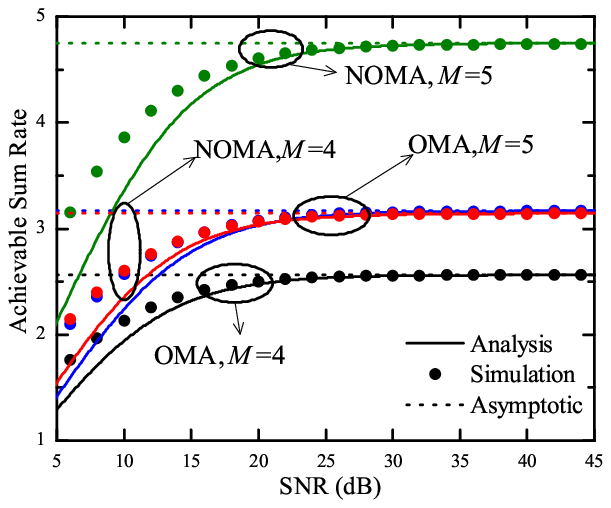}
\end{minipage}
}
\subfigure[{Ideal or non-ideal conditions}]{
\begin{minipage}[t]{0.465\linewidth}
\centering
\includegraphics[width= 1.8in, height=2.3in]{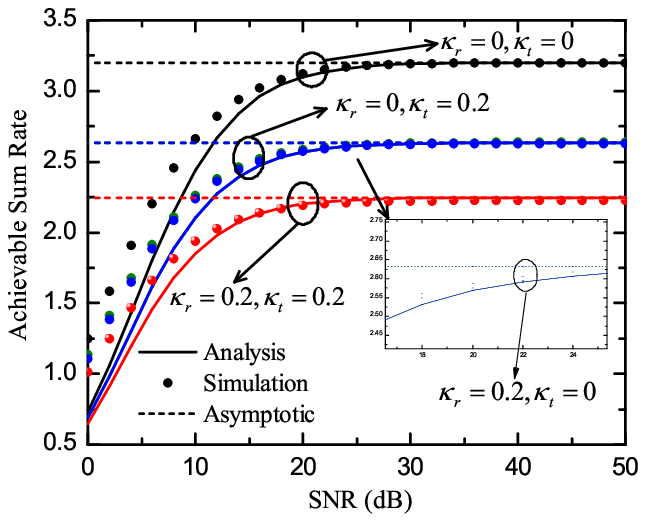}
\end{minipage}
}
\caption{The achievable sum rate of users.}
\end{figure}
\section{Numerical Results \& Discussion}\label{sec5}
In this section, indicative numerical results are provided to verify the correctness of our theoretical analysis and verified by using Monte Carlo simulations with $10^5$ trials. Unless otherwise stated, we use the following parameter settings \cite{6630485Bjornson}: $M = 3$, ${\sigma _r} = {\sigma _{t, k}} = 1$, ${\kappa _{{u_{t,i}}}} = {\kappa _{Rt}} = {\kappa _{Rr}} = {\kappa _{{u_{r,k}}}} = 0$, $d_i = 1$\footnote{This distance is the normalized distance.}, $\nu  = 3$, ${a_1} = 0.5,{a_2} = 0.3,{a_3} = 0.2$ and ${\alpha _i} = 2,{\beta _i} = 3$, $P = nP_R$, where $n$ is a fixed constant.

Fig. 2(a) plots the ASR versus the average SNR for difference number of users $M = \{4, 5\}$. For $M = 4$, the power allocation coefficients are set to ${a_1} = 0.5$, ${a_2} = 0.3$, ${a_3} = 0.15$, ${a_4} = 0.05$. For $M = 5$, the power allocation coefficients are set to ${a_1} = 0.5$, ${a_2} = 0.2$, ${a_3} = 0.15$, ${a_4} = 0.1$, ${a_5} = 0.05$. Moreover, for comparsion purpose, the curves of ASR with OMA are also provided \cite{C.D.Ho2018}. From Fig. 2(a), we can observe that the ASR in both NOMA and OMA grows larger when the number of users increases. This is because that the ASR is contributed by all users. Another observation is that the ASR of NOMA higher than that of OMA. The reason is that the proposed system can reduce time-slot to two from $\left[ {\left( {M - 1} \right)/2} \right] + 1$ in OMA.  Finally, the ASR at high SNRs approaches to a fixed constant due to the large interference, which illustrates the Remark 2 we have obtained.

\begin{figure}[!t]
\setlength{\abovecaptionskip}{0pt}
\centering
\includegraphics [width=3.7in]{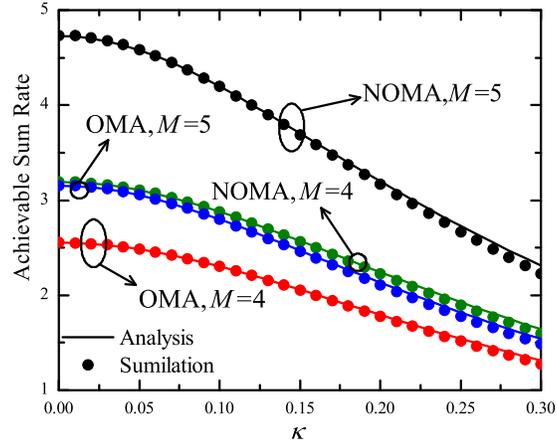}
\caption{The achievable sum rate vs. $\kappa$ for OMA/NOMA schemes.}
\label{fig4}
\end{figure}

\begin{figure}[t!]
\centering
\subfigure[{Different SNR}]{
\begin{minipage}[t]{0.46\linewidth}
\centering
\includegraphics[width= 1.8in, height=2.3in]{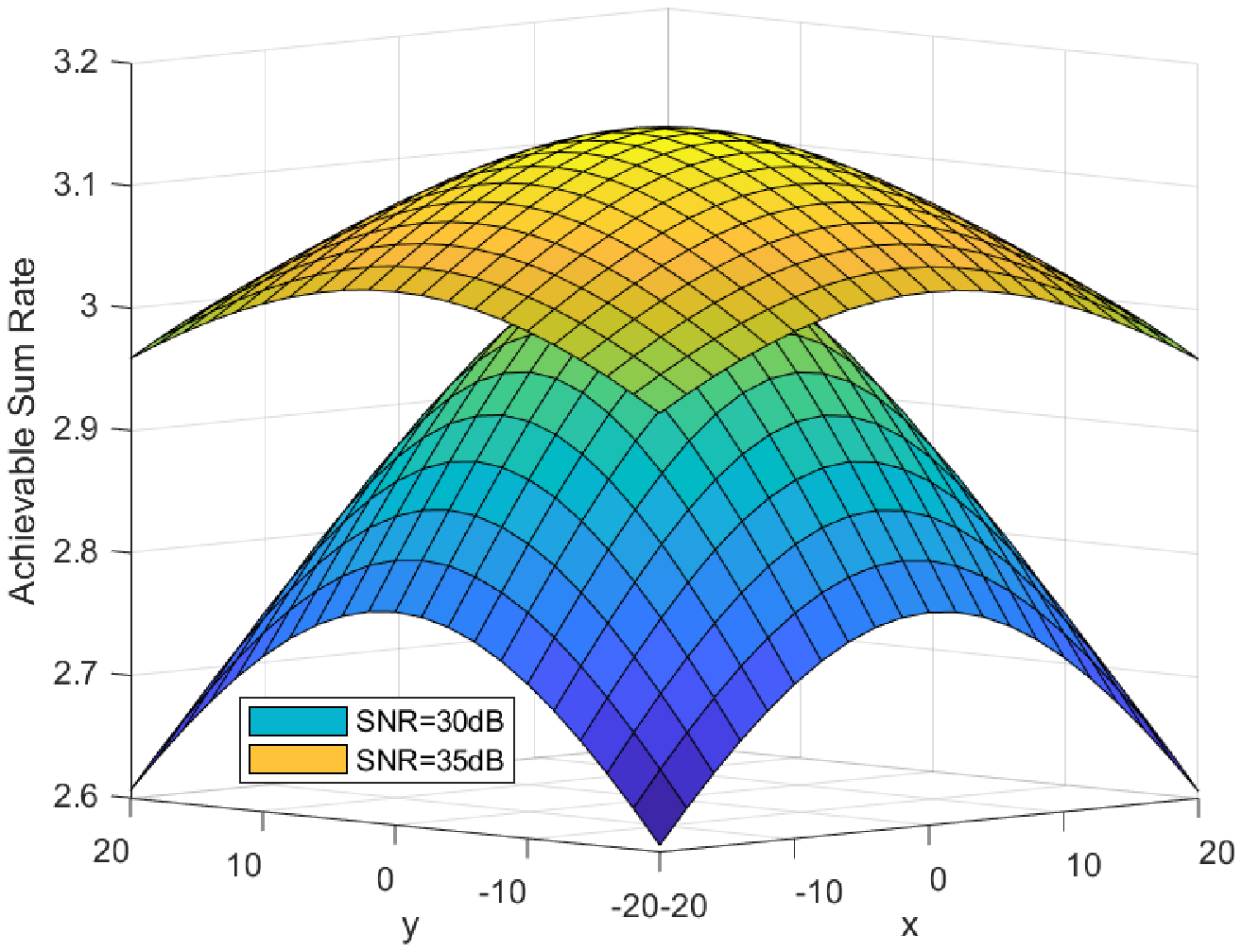}
\end{minipage}
}
\subfigure[{OMA/NOMA}]{
\begin{minipage}[t]{0.46\linewidth}
\centering
\includegraphics[width= 1.8in, height=2.3in]{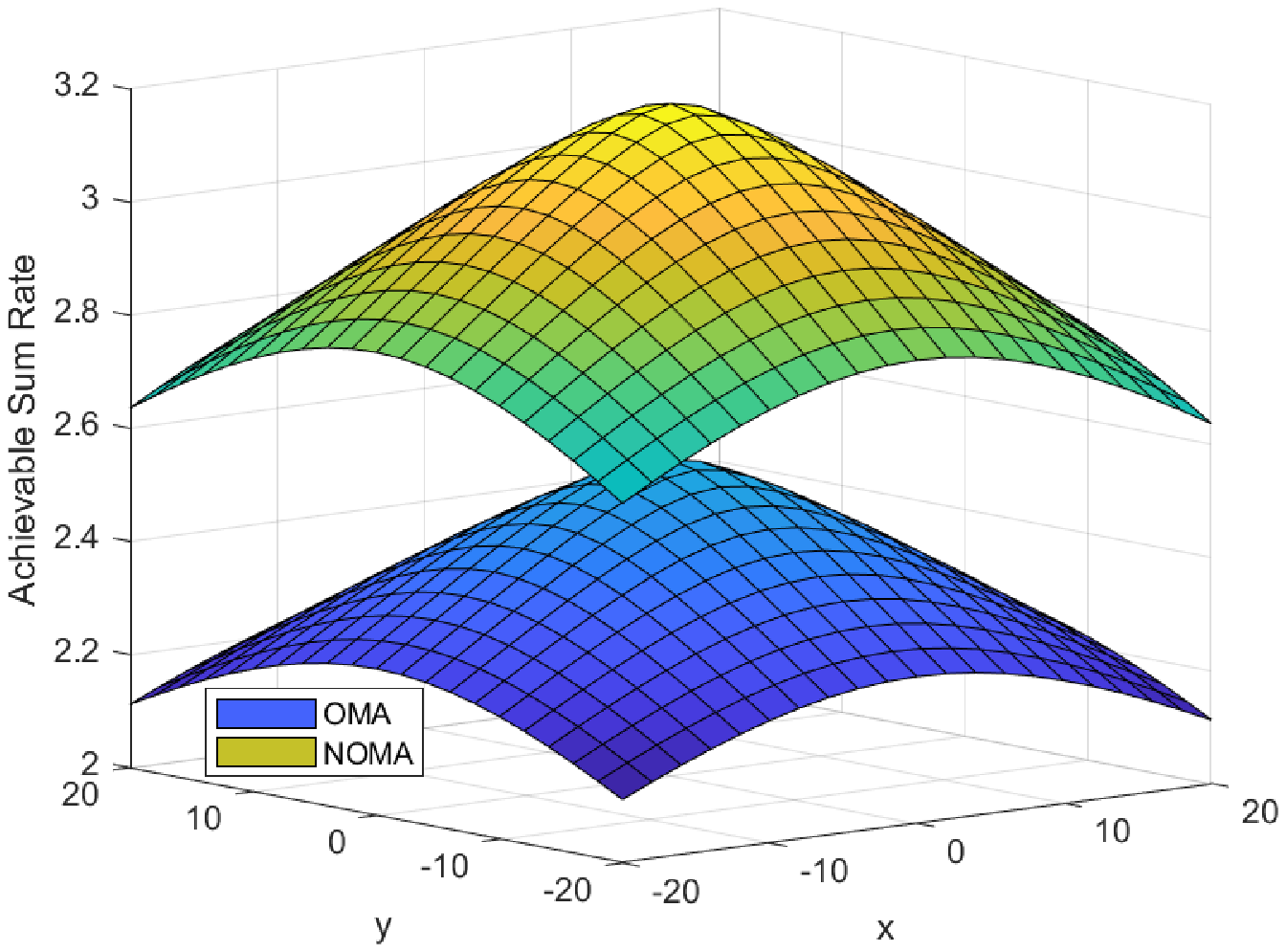}
\end{minipage}
}
\caption{The achievable sum rate of users.}
\end{figure}
Fig. 2(b) depicts the ASR versus the average SNR under ideal and non-ideal conditions with four users. To characterize the effects of RHIs at the transmitter, receiver and transceiver, we consider three types of parameter settings: 1) The transmitter RHIs, (${\kappa _{{u_{t,i}}}} = {\kappa _{Rt}} =0, {\kappa _{Rr}} = {\kappa _{{u_{r,k}}}} = 0.2$); 2) The received RHIs, (${\kappa _{{u_{t,i}}}} = {\kappa _{Rt}} = 0.2,{\kappa _{Rr}} = {\kappa _{{u_{r,k}}}} = 0$); and 3) The transceiver RHIs, (${\kappa _{{u_{t,i}}}} = {\kappa _{Rt}} = {\kappa _{Rr}} = {\kappa _{{u_{r,k}}}} = 0.2$). From Fig. 2(b), it can be observed that the ASR under ideal conditions outperforms non-ideal conditions. More significantly, the gap of the ASR under the condition of transmitter RHIs and receiver RHIs is almost neglectful. This means that they have the same effects on the ASR. One can also observe that RHIs at transceiver has serious effects on the ASR.

Fig. 3 illustrates the ASR variation trend versus level of RHIs for different number of user $M = \{4, 5\}$. In this simulation, we assume SNR is 30dB. Similarly, the curves of ASR under case of OMA are provided. As can be seen from Fig. 3 that all the ASRs become smaller as the value of parameters of RHIs $\kappa$ grows large. This means that the ASR of considered sytems is limited by RHIs. In addition, the gap of ASR between NOMA and OMA becomes large as the number of users increases.

Fig. 4(a) shows the achievable sum rate of NOMA users versus the location of users when the SNR is set to 30dB and 35dB. The  height of UAV is fixed at 10m and the numbers of users are 4 in this case. One is obvious that the distance from the UAV to the ground NOMA users can affect the magnitude of the ASR. It is also worth to mention that when $\left( {x,y} \right) \to \left( {0,0} \right)$, the ASR of the NOMA users grows large due to the small average distance between the UAV and the served NOMA users. This means that we can deploy the UAV to maximize the ASR of the considered networks. In addition, the ASR when SNR is 35dB is greater than when SNR is 30dB. Fig. 4(b) describes the achievable sum rate of NOMA users and OMA users versus the location of users. The SNR is set to 30dB and the numbers of users are 4 as in Fig. 4(b). Obviously, the ASR of UAV-aided NOMA MWRN is better than that of OMA due to the short communication time.

\section{Conclusion}
This paper studied the ASR of UAV-aided NOMA MWRNs over Nakagami-$m$ fading in the presence of RHIs at transceivers. Specifically, we derived the analytical closed-form expression for the ASR and the asymptotic analysis in the high SNR region. In addition, the \emph{high SNR slope} and \emph{high SNR power offset} were discussed. The simulation results showed that the system performance of the considered networks outperforms traditional OMA MWRNs. In addition, we concluded that the ASR of the UAV-Aided NOMA MWRNs is limited by RHIs at transceiver. Finally, we identified that UAV can be flexible deploying to optimize ASR performance.


\bibliographystyle{IEEEtran}
\bibliography{mybib}
\end{document}